\documentclass[12pt,a4paper]{article}

\usepackage{amsmath}
\usepackage{amstext}
\usepackage{amssymb}
\usepackage{amsopn}
\usepackage{amsthm}
\usepackage{braket}
\usepackage{mathtools}
\usepackage{braket}
\usepackage{hyperref}

\def\HH{\mathcal{H}}
\def\EE{\mathcal{E}}

\DeclareMathOperator{\tr}{Tr}

\newcommand{\kket}[1]{\vert #1 \rangle \rangle}
\newcommand{\bbra}[1]{\langle \langle #1 \vert}
\newtheorem{definition}{Definition}
\newtheorem{proposition}{Proposition}

\title{Choi representation of completely positive maps in brief}

\author{G\'abor Homa\thanks{HUN-REN Wigner Research Centre for Physics, P.~O.~Box 49, H-1525 Budapest, Hungary}
 \and Adrian Ortega\thanks{Departamento de F\'isica, Universidad de Guadalajara, Blvd. Gral. Marcelino Garc\'ia Barrag\'an 1421, C.P. 44430, Guadalajara, Jalisco, Mexico and HUN-REN Wigner Research Centre for Physics, P.~O.~Box 49, H-1525 Budapest, Hungary}
\and M\'aty\'as Koniorczyk\thanks{HUN-REN Wigner Research Centre for Physics, P.~O.~Box 49, H-1525 Budapest, Hungary}}

\begin{document}

\maketitle

\begin{abstract}
  The Choi representation of completely positive (CP) maps, i.e. quantum
  channels is often used in the context of quantum information and computation as it is easy to work with. It is a correspondence between CP maps and quantum states also termed as the Choi-Jamio\l kowski isomorphism. It is especially useful if a parametrization of the set of CP maps is needed in order to consider a general map or optimize over the set of these.  Here we provide a brief introduction to this topic, focusing on certain useful calculational techniques which are presented in full detail.
\end{abstract}



Completely positive trace preserving maps are the
mathematical models of most general physical processes in quantum
mechanics: they describe the transformation of a density operator
describing the state of the system to another density operator
describing the state after any physically realistic process.

The mathematical foundations of quantum mechanics had been laid down by
von Neumann~\cite{Neumannbook1955}. Open quantum systems are important
subjects of quantum mechanics (c.f. e.g.\cite{Daviesbook1976}), and
especially in this theory it is a central question what change a
quantum state can undergo most generally, even when it is coupled to
its environment.  The seminal book of Kraus~\cite{Krausbook1983} is an
important reference with this respect, in which the formalism
describing quantum channels acquires a prominent role.

The importance of the topic can be further underlined by a number of
other references. For instance, in quantum communication theory, they
model the most general change to a quantum state that it can undergo
when passing a noisy or distorting channel. Quantum channels have a
wide reach in fields such as Quantum Computation and Quantum
Information~\cite{Nielsenbook2010, Benentibook2019}, Open Quantum
Systems~\cite{Breuerbook2007}, Quantum Measurement and
Control~\cite{Wisemanbook2009}, and is also a subject of mathematical
interest~\cite{Heinosaaribook2012, Watrousbook2018, Holevobook2019}.

The present brief tutorial has a rather limited and technical scope.
Namely, we assume that the reader is familiar with the Kraus
representation of completely positive maps~\cite{Krausbook1983}: this
can be considered as basic knowledge in the field. On this basis we
derive the Choi representation of completely positive
maps~\cite{Choi_1975} (aka Choi-Jamio\l kowski isomorphism) on a very technical level, providing all details
of the elementary calculations. Our goal is to help the reader to do
calculations with quantum channels in Choi representation. The main
benefit when compared to the Kraus representation is that the entire
set of completely positive maps can be covered by a convenient
parametrization, which can have benefits when considering general
channels or optimizing over their sets. The present notes are mainly
intended for those who address such a technical task.  We do not go
into deeper discussions of the Choi-Jamio\l kowski
isomorphism~\cite{Choi_1975,Jamiolkowski_1972}; there are excellent
deeper discussions of the topic in the literature, including recent
contributions, e.g.~\cite{2211.16533}.


It can be used as a material in a lecture on Quantum Information for
advanced undergraduates in Physics or Mathematics or recently graduated
students; the discussion only requires as a background Linear Algebra,
Quantum Mechanics at a basic level, and some familiarity with Quantum
Channels. As such, it can also be used for self-study, as the
manuscript is self-contained. Yet, it is expected that the reader
participates actively when reading the manuscript, following closely
all the derivations we present. For applications and in order to
appreciate the material presented here, the reader is referred to the
citations written in this introduction, as well as, to the couple of
applications we have added at the end of this manuscript.

\section{Notation, preliminaries}

Let $\HH _1$ and $\HH _2$ be finite-dimensional complex
Hilbert-spaces.  Let $\varrho$ denote a (Hermitian)
positive semidefinite operator of unit trace. (In physics,
such a $\varrho$ represent the most general state of a physical system,
whereas $\EE$ represent one of the channels a state can pass through.) The set of bounded operators mapping a Hilbert space $\HH$ into itself will be denoted by $ \mathcal{B}(\HH)$.

Further, let $\EE$ be a $\mathcal{B}(\HH_ 1) \mapsto \mathcal{B}(\HH_2)$ \emph{completely
  positive trace preserving} (CPTP) map. We assume that the reader is
aware that $\EE$ is CPTP iff it can be written in Kraus
representation: for every $\varrho: \HH _1 \mapsto \HH _1$ (Hermitian)
positive semidefinite operator,
\begin{eqnarray}
  \label{eq:Kraus}
  \EE (\varrho) = \sum_k A^{(k)} \varrho A^{(k)\dag}\\
  \sum_k A^{(k)\dag} A^{(k)} = \hat 1 _1,\label{eq:Kraus2}
\end{eqnarray}
where $A^{(k)}$-s are $\HH _1 \mapsto \HH _2$ linear operators, their
number can be arbitrary.
The Kraus-representation of a CPTP map is not
unique; we will return to this, to some extent, later on. 

Note that it is easy to see that the condition in Eq.~\eqref{eq:Kraus2} is a necessary and sufficient condition for the map to be trace preserving. Indeed, using the linearity and the cyclic property of trace, and the associative property of matrix product,
\begin{eqnarray}
    \tr \EE (\varrho) = 
    \tr \left(\sum_k A^{(k)} \varrho A^{(k)\dag}\right) =
    \sum_k \tr \left( A^{(k)} \varrho A^{(k)\dag}\right)=\nonumber \\
    \sum_k \tr \left( A^{(k)\dag} A^{(k)} \varrho\right)=
     \tr \left( \left(\sum_k
     A^{(k)\dag} 
     A^{(k)} \right) \varrho\right).    
\end{eqnarray}
For this to be equal to $\tr(\varrho)$ for all $\varrho$, 
it is necessary to have $\sum_k A^{(k)\dag} A^{(k)} = \hat 1_1$, which is Eq.~\eqref{eq:Kraus2}.

We will use the Dirac notation for vectors, e.g.
$\ket{\psi}_1 \in \HH _1$, and the dual of $\ket{\psi}_1$ is
$\prescript{}{1}{\bra{\psi}}\in \HH _1'$. When $\HH _1$ is represented with coordinates
on an orthonormal basis (ONB) in $\mathbb{C}^N$, then $\ket{\psi}_1$ is a column vector, and $\prescript{}{1}{\bra{\psi}}=\ket{\psi}_1^\dag$ is its Hermitian conjugate: a row vector, which is, in a complex vector space, the complex conjugate of the transpose.
Further, let $\{\ket{j}\}_{j=1,\ldots, \dim \HH _1}$ and $\{\ket{i}\}_{i=1,\ldots, \dim \HH _2}$ be ONBs on $\HH _1$ and $\HH _2$, respectively.

First we define a vectorization of linear mappings and present some properties of this vectorization that will be useful later.
\begin{definition}
  \label{def:kett}
  For a linear operator $C: \HH _1 \mapsto \HH _2$, with the matrix
  \begin{equation}
    \label{eq:ChoivecInv}
    C_{i,j}=\prescript{}{2}{\bra{i}}C\ket{j}_1, 
   \end{equation}
  and for the ONBs $\{\ket{j}\}_{j=1,\ldots, \dim \HH _1}$ and $\{\ket{i}\}_{i=1,\ldots, \dim \HH _2}$, let
  \begin{equation}
    \label{eq:Choivec}
    \kket{C} := \sum_{i,j} C_{i,j} \ket{j}_1 \ket{i}_2:= \sum_{i,j} C_{i,j} \ket{j,i} \in \HH _1 \otimes \HH _2 
  \end{equation}
  be a vector in the product of the two spaces.  Here $\ket{j}_1 \ket{i}_2=\ket{j,i}$   stands for $\ket{j}_1 \otimes \ket{i}_2$, elements of an ONB
  spanning the product space.
\end{definition}
(Note that the "double ket" ($\kket{\dots}$) notation is common in the literature to denote vectorized operators, appearing in many works including early ones, c.f.~\cite{Royer_1991}.) The Hermitian conjugate of this vector is, by definition,
\begin{equation}
    \bbra{C} := \left(\kket{C}\right)^\dag=\sum_{i,j} \prescript{}{2}{\bra{i}} \prescript{}{1}{\bra{j}} C^{\ast}_{i,j} := \sum_{i,j} \bra{i,j} C^{\ast}_{i,j}  \in \left(\HH _1 \otimes \HH _2\right)'.
\end{equation}
That is, on an ONB the Hermitian conjugate is the transpose of the complex conjugate (the asterisk ($\mbox{}^\ast$) stands for complex conjugation). As it is a vector, the order of the indices of $C$ remain the same. Note the reverse order of the indices in the bra, i.e. the dual vector, which will be our convention. 
\begin{proposition}
  For the matrices of all linear operators
  \begin{align*}
    A:& \;\HH_1 \mapsto \HH_1 \\
    B:& \;\HH_2 \mapsto \HH_2 \\
    C:& \;\HH _1 \mapsto \HH _2,
  \end{align*}
  on the basis in which the vectorization in Def. \ref{def:kett} is defined,
  the identity
  \begin{equation}
    \label{eq:prop1stmt}
    (A\otimes B)\kket{C} = \kket{BCA^\top}
  \end{equation}
  holds.
\end{proposition}
Recall that transposition is not basis invariant in complex spaces.
\begin{proof}
  We verify each coordinate of the vector on the two sides to be equal.
   By definition \ref{def:kett}, and with some algebra (note that we keep on using the same bases):
  \begin{eqnarray*}
    \forall (i,j):\\
     \bra{i,j} \left( (A\otimes B)|C\rangle\rangle \right) =
    \bra{i,j}\left( (A\otimes B) \sum_{k,l} C_{k,l}\ket{l,k} \right)=\\
    \sum_{k,l} C_{k,l} \prescript{}{1}{\bra{j}}A\ket{l}_1 \prescript{}{2}{\bra{i}} B\ket{k}_2 = \sum_{k,l} C_{k,l} A_{j,l} B_{i,k}=\\
    \sum_{k,l} B_{i,k} C_{k,l} A^\top_{l,j} = \left(BCA^\top\right)_{i,j}.
  \end{eqnarray*}
\end{proof}

For sake of consistency we describe here the notion of partial trace.
\begin{definition}
  The partial traces for linear operator
  $M: \HH_1\otimes \HH_2 \mapsto \HH_1\otimes \HH_2$
   are defined as
  \begin{equation}
    \label{eq:ptrace2}
    \prescript{}{1}{\bra{i}} \left(M^{(1)}\right) \ket{j}_1 = \left(\tr_2 M \right)_{i,j} =
      \sum_k  \prescript{}{2}{\bra{k}}\prescript{}{1}{\bra{i}} M \ket{j}_1\ket{k}_2 =M^{(1)}_{i,j},
    \end{equation}
    that is, $M^{(1)}$ is a $\HH_1 \mapsto \HH_1$ linear operator with the matrix
    \begin{equation}
      \label{eq:ptrace2detail}
      M^{(1)}_{i,j}=\sum_k M_{ik,jk}.
    \end{equation}
    We say $M^{(1)}$ is obtained by \emph{tracing out} $M$ on $\HH_2$. Similarly we have
    $M^{(2)}=\tr_1 M$ which is a $\HH_2 \mapsto \HH_2$ so that
    \begin{equation}
      \label{eq:ptrace2detail2}
      M^{(2)}_{k,l}=\sum_i M_{ik,il}.
    \end{equation}
\end{definition}
    (The index pairs not separated with a comma, e.g. $ik$ in Eq.~\eqref{eq:ptrace2detail} emphasize that they belong to an element of a direct product basis, while the comma separation in $M$'s indices stand here for to separate matrix indices. We use this convention whenever the matrix interpretation is relevant, albeit it is not always possible to do it consistently.)
Just like the 'complete' trace, it is basis-independent, because the
partial trace is a unique linear map such that
$\tr_2(A\otimes B)=A\,\mathrm{tr}(B)$, and the trace is
basis-independent. The partial trace is often used in
physics. Notably, given two Hermitian operators $M_A: \HH_1 \mapsto \HH_1$
and $M_B: \HH_2 \mapsto \HH_2$ both with unit trace, 
\begin{eqnarray}
  \label{eq:prodinv}
  \tr_2 (M_A \otimes M_B) = M_A  \nonumber \\
  \tr_1 (M_A \otimes M_B) = M_B  
\end{eqnarray}
holds.  Hence, for Hermitian
operators the partial trace is the ``inverse'' of the direct product. 

\begin{proposition}
\label{prop_2}
  $\forall C_1,C_2 : \HH _1 \mapsto \HH _2$ pairs of linear operators,
  \begin{equation}
    \label{eq:prop2stmt}
    \tr_1 \left( \kket{C_1}\bbra{C_2} \right) = C_1 C_2^\dag
  \end{equation}
  holds.
\end{proposition}
\begin{proof}
  By the definition of our vectorization, 
\begin{eqnarray}
    \label{eq:thm2proof}
    \left(\forall (k,l)\right)\quad
    \left(\tr_1 \left( \kket{C_1}\bbra{C_2}\right)\right)_{k,l} =
    \sum_j \left( \kket{C_1}\bbra{C_2}\right)_{jk,jl}\nonumber \\
    \sum_j \left( \prescript{}{2}{\bra{k}}\prescript{}{1}{\bra{j}}
    \kket{C_1}\right)\left(\bbra{C_2} 
    \ket{j}_1\ket{l}_2 \right)
    = \sum_j (C_1)_{k,j}(C^\ast_2)_{l,j}=
    \nonumber \\
    \sum_j (C_1)_{k,j}(C_2^\dag)_{j,l}=
    \left(  C_1 C_2^\dag\right)_{k,l}.
  \end{eqnarray} 
\end{proof}

\section{Definition and some properties of the Choi representation}

After these preliminaries now let us define the Choi matrix of a completely positive map.

\begin{definition}
  \label{def:choi}
 For $\EE$ operator, a complete positive $\mathcal{B}(\HH_ 1) \mapsto \mathcal{B}(\HH _2)$, given in a
Kraus representation (c.f. Eq.~\eqref{eq:Kraus}), its Choi matrix is the
following Hermitian $\HH_1\otimes \HH_2 \mapsto \HH_1 \otimes \HH_2$ matrix:
\begin{equation}
  \label{eq:Choidef}
  X_\EE = \sum_k \kket{A^{(k)}}\bbra{A^{(k)}}.
\end{equation} 
\end{definition}
(Note that $X_\EE$ is Hermitian as it is a sum of rank-one projectors.)

One of the convenient features of the Choi representation is that it is easy to prove the following statement.

\begin{proposition}
\label{prop_3}
  $\forall \varrho \in \mathcal{B}(\HH_1)$ operators,
  \begin{equation}
   \label{eq:choiact}
    \EE(\varrho)=\tr_1 \left( \left(\varrho^\top\otimes \hat 1_2\right) X_\EE \right)    \in \mathcal{B}(\HH_2).
  \end{equation}
\end{proposition}
\begin{proof}
  By the definition in Eq.~\eqref{eq:Choidef} and the linearity of partial trace we have
  \begin{equation}
    \EE(\varrho)=\tr_1 \left( \left(\varrho^\top\otimes \hat 1_2\right) X_\EE \right) =
    \sum_k \tr_1\left( \left(\varrho^\top\otimes \hat 1_2\right)\kket{A^{(k)}}\bbra{A^{(k)}}  \right).
  \end{equation}
  Using Eq.~\eqref{eq:prop1stmt}, this can be written as
  \begin{equation}
    \sum_k \tr_1\left( \left(\varrho^\top\otimes \hat 1_2\right)\kket{A^{(k)}}\bbra{A^{(k)}}  \right)=
    \sum_k \tr_1\left( \kket{\hat 1_2A^{(k)}\varrho}\bbra{A^{(k)}}\right),
  \end{equation}
  which, using Eq.~\eqref{eq:prop2stmt} and linearity, gives
  \begin{equation}
    \sum_k \tr_1\left( \kket{\hat 1_2A^{(k)}\varrho}\bbra{A^{(k)}}\right)=\sum_k A^{(k)}\varrho A^{(k)\dag}=\EE(\varrho).
  \end{equation}
\end{proof}

Given a CPTP map $\EE$, it is possible to give an alternative definition of its Choi matrix, which is in fact its original definition~\cite{Choi_1975}:
\begin{definition}
  \begin{equation}
X_\EE:=\sum_{p,q} E^{p,q}  \otimes \EE(E^{p,q})  \label{eq:Choialtdef}
\end{equation}
where
\begin{equation}
  \label{eq:E}
  E^{p,q}: \ \HH^1 \mapsto \HH^1,\qquad
  \left(E^{p,q}\right)_{j,m}=\delta_{j,p}\delta_{m,q}.
\end{equation}
\end{definition}
\begin{proposition}
 Definitions in Eq.~\eqref{eq:Choidef} and Eq.~\eqref{eq:Choialtdef} are equivalent.
\end{proposition}
\begin{proof}
Let us consider the Kraus represention of $\EE$.
  By the definition of our vectorization in Eq.~\eqref{eq:Choivec} we have, $\forall (i,j,l,m)$
  \begin{eqnarray}
      \prescript{}{2}{\bra{i}}\prescript{}{1}{\bra{j}} \left( \sum_k \kket{A^{(k)}}\bbra{A^{(k)}}\right) \ket{m}_1\ket{l}_2=\nonumber \\
     \sum_k \left( \prescript{}{2}{\bra{i}}\prescript{}{1}{\bra{j}} \kket{A^{(k)}} \right) 
     \left( \bbra{A^{(k)}} \ket{m}_1\ket{l}_2\right)
    =\sum_k A^{(k)}_{i,j}  A^{(k)\ast}_{l,m}.
  \end{eqnarray}
  On the other hand, 
  \begin{eqnarray}
    \prescript{}{2}{\bra{i}}\prescript{}{1}{\bra{j}} \left(
    \sum_{p,q} E^{p,q} \otimes \EE(E^{p,q})    \right) \ket{m}_1\ket{l}_2=
    \sum_{p,q} \left(\prescript{}{1}{\bra{j}}E^{p,q}\ket{m}_1\right) 
    \left(\prescript{}{2}{\bra{i}}\EE(E^{p,q})\ket{l}_2\right)\nonumber \\
    = \sum_{p,q} \delta_{j,p}\delta_{m,q}
    \prescript{}{2}{\bra{i}}\EE(E^{p,q})\ket{l}_2= 
    \prescript{}{2}{\bra{i}}\EE(E^{j,m})\ket{l}_2=
    \EE(E^{j,m})_{i,l} \nonumber \\
    \left(\sum_k A^{(k)} E^{j,m}A^{(k)\dag}  \right)_{i,l}= \sum_{k,r,s} 
    A^{(k)}_{i,r} E^{j,m}_{r,s}A^{(k)\dag}_{s,l}=
     \sum_{k,r,s} 
    A^{(k)}_{i,r} \delta_{j,r}\delta_{m,s} A^{(k)\dag}_{s,l} = \sum_k A^{(k)}_{i,j}  A^{(k)\ast}_{l,m}. \nonumber \\
      \end{eqnarray}
\end{proof}

It can be shown that the Choi matrix of a CPTP map as defined in
Definition~\ref{def:choi} is positive semidefinite. And the converse
is also true: an arbitrary Hermitian positive semidefinite
operator $X$ represents a CPTP map as its Choi matrix in the sense 
that $\EE(\varrho)$ can be calculated according to
Eq.~\eqref{eq:choiact}.  This follows from the following two
propositions.

\begin{proposition}
\label{prop_5}
The Choi matrix $X_{\EE}$ of a CPTP map as defined in Eq.~\eqref{eq:Choidef} is positive semidefinite.
\end{proposition}

\begin{proof}
Consider a spectral decomposition of $X_{\EE}$
\begin{equation}
\label{eq: Choi_decomposition}
X_{\EE}  = \sum_k\lambda_k \ket{\Psi_k}\bra{\Psi_k},
\end{equation}
where $(\ket{\Psi_k})$ is an ONB and $\lambda_k$ are the eigenvalues
of operator $X_{\EE}$ (degenerate ones are repeated respectively). We
have to show that $\lambda_k\geq 0$:
\begin{eqnarray}
\lambda_k= \bra{\Psi_k} X_\EE\ket{\Psi_k}  =  \sum_k \bra{\Psi_k} \left(\kket{A^{(k)}}\bbra{A^{(k)}} \right)\ket{\Psi_k}  \nonumber \\
=  \sum_k  \bra{\Psi_k} \sum_{i,j} A_{i,j}^{(k)} \ket{j,i} \bra{i,j} \sum_{i,j}  A_{i,j}^{(k)\dag} \ket{\Psi_k}   \nonumber \\
=  \sum_{k,i,j} \left| \bra{\Psi_k} A^{(k)}_{i,j} \ket{i,j} \right|^2  \geq 0.
 \end{eqnarray}
\end{proof}

Assume now that we are given a Hermitian positive semidefinite matrix
which is supposed to be a Choi matrix of a completely positive trace
preserving map according to Eq.\eqref{eq:choiact}. We will prove the
complete positivity later, but first we check the condition for being
trace preserving.

\begin{proposition}
\label{prop_6}
  Given $X\geq 0: \HH_1 \otimes \HH_2 \mapsto \HH_1 \otimes \HH_2$
  Hermitian matrix, it is the Choi representation of a trace
  preserving map iff
  \begin{equation}
    \label{eq:tpcond}
    \tr_2 X =  \hat 1 _1.
  \end{equation}
  \begin{proof}
  Let $\EE$ be the corresponding CP map according to Eq.~\eqref{eq:choiact}.
  and $|n,m\rangle$ a complete ONB. Then we have
   \begin{eqnarray}
      \tr \EE(\varrho) = \tr \left(\tr_1 \left( \left(\varrho^\top\otimes \hat 1_2\right) X_\EE \right) \right)
    = \sum_{i,j} \left( 
       \left( \varrho^\top \otimes \hat 1_2  \right) X_\EE
    \right)_{ij,ij} =
    \sum_{i,j,k,l} 
       \left( \varrho^\top \otimes \hat 1_2  \right)_{ij,kl} X_{\EE\, kl,ij}
      \nonumber \\
      = \sum_{i,j,k,l} \varrho^\top_{i,k}\delta_{j,l}X_{kl,ij}=\sum_{i,j,k,l} \varrho_{k,i} \delta_{j,l}X_{kl,ij}
      = \sum_{i,k,l} \varrho_{k,i} X_{kl,il} = \sum_{i,k} \varrho_{k,i}\left(\tr_2 (X)\right)_{k,i}.
      \label{eq:tpcondproof}
    \end{eqnarray}
    This is equal to $\tr(\varrho)$ iff $\left(\tr_2 (X)\right)_{k,i}=\delta_{k,i}$, that is, $\tr_2 X = \hat 1_1$.
  \end{proof}
\end{proposition}

\begin{proposition}
  For a Hermitian positive semidefinite
  $\HH_1\otimes \HH_2 \mapsto \HH_1 \otimes \HH_2$ operator $X$ with
  property $\tr_2 X =  \hat 1 _1$, the mapping $\EE$ defined using Eq.~\eqref{eq:choiact}
  obeys a Kraus representation.
\end{proposition}

\begin{proof}
  Let $X$ an arbitrary Hermitian positive semidefinite operator with property $\tr_2 X =  \hat 1 _1$.
  Let us write it in a spectral decomposition
  \begin{equation}
    \label{eq:xspectr}
    X  = \sum_j\kappa_k \ket{\Phi_k}\bra{\Phi_k}
  \end{equation}
  where
  \begin{equation}
    X\ket{\Phi_k}=\kappa_k\ket{\Phi_k}.
  \end{equation}
  As $X$ is Hermitian, the eigenvalues $\kappa_k$ are real, and the
  eigenvectors $\left( \ket{\Phi_k} \right)_K$ can be chosen so that
  they form an ONB. (Note that the spectral decomposition in
  Eq.~\eqref{eq:xspectr} is unique up to a unitary transformation in
  the degenerate subspaces.
  Now let us set
  \begin{eqnarray}
    \label{eq:Binvdef}
    \kket{A^{(k)}} := \sqrt{\kappa_k} \ket{\Phi_k}.
   \end{eqnarray}
   Let us now fix an orthonormal product basis $\ket{j}_1 \ket{i}_2$,
   and define the operators $A^{(k)}$ with their matrix on this ONB by
   inverting the vectorization in Definition~\ref{def:kett}, that is,
   according to Eqs.~\eqref{eq:Choivec} and~\eqref{eq:ChoivecInv}:
  \begin{equation}
     A^{(k)}_{i,j}=\prescript{}{2}{\bra{i}}\prescript{}{1}{\bra{j}}\kket{A^{(k)}}.
     \label{eq:ChoiInvA}
   \end{equation}
   For the so-defined operators, when considered as Kraus operators of
   a CP-map $\mathcal{E}$ as defined in Eq.~\eqref{eq:Kraus},
   Eq.~\eqref{eq:choiact} will obviously hold with $X=X_\mathcal{E}$,
   according to the proof of Proposition~\ref{prop_3}.

   It remains to show that for the operators defined in
   Eq.~\eqref{eq:ChoiInvA}, the condition Eq.~\eqref{eq:Kraus2} will also hold. This is indeed the case as the condition is the necessary and sufficient condition in the Kraus representation for the map to be trace preserving. Because of our prescription $\tr_2 X = \hat 1 _1$ the map is trace preserving, hence, the condition in Eq.~\eqref{eq:Kraus2} holds.

 \end{proof}

 The previous proof has a consequence on the question of the ambiguity
 of the Kraus representation.  Namely, the construction defined by
 Eqs.~\eqref{eq:Binvdef} and~\eqref{eq:ChoiInvA} define an
 \emph{orthogonal Kraus representation}, which is unique up to the
 possible choice of the eigenvectors in the degenerate eigensubspaces.
 It is easy to show that for the Kraus operators defined in Eqs.~\eqref{eq:Binvdef}
 and~\eqref{eq:ChoiInvA},
 \begin{equation}
   \label{eq:ortogkrauss}
   \tr  A^{(k)\dag} A^{(l)} \propto \delta_{k,l}\hat 1
 \end{equation}
 holds. This also implies that
 the Choi operator is unique: each positive semidefinite Hermitian
 operator  $X$ with $\tr_2 X = \hat 1 _1$ defines a unique CPTP map.

The convenience of the Choi representation lies in the fact that it
describes completely positive maps with Hermitian matrices (compliant
with the condition in Eq.~\eqref{eq:tpcondproof} to be trace
preserving). Recall that the CP maps are linear mappings but writing
them in a matrix form (aka left-right representation) helps little
with e.g. parametrizing them. The Kraus representation reveals the
structure but as it consists of arbitrary operators, again it is
rather hard to parametrize.

In addition the Kraus representation is ambiguous. In fact this can be
better understood on the basis of the Choi representation. When
writing the Choi matrix in diagonal form, we can identify a particular
Kraus representation built up from its orthogonal eigenvectors. This
orthogonal Kraus representation is unique up to unitary
transformations in possible degenerate eigensubspaces. All other Kraus
representations of the same CP map arise from writing the Choi matrix
in other bases.

\section{Example: the qubit amplitude damping channel}

Let us provide an elementary example of a typical qubit channel's Choi
matrix as derived from its Kraus representation, and the use for to
calculate the operation of the channel. The calculation described here
is available as a code for the Maxima open-source computer algebra
system~\cite{maxima} as a supplementary material to this article.

The qubit amplitude damping channel has a parameter $p\in[0,1]$, and maps a qubit density matrix written in a qubit basis $\ket{0}, \ket{1}$ as 
\begin{equation}
\varrho^{(\text{in})}=
\begin{pmatrix}\varrho_{0,0} & \varrho_{0,1} \\
\varrho_{1,0} & \varrho_{1,1}
    \end{pmatrix}
\end{equation}
to
\begin{equation}
  \label{eq:adchannel}
  \varrho^{(\text{out})}=
  \begin{pmatrix}p \varrho_{1,1}+\varrho_{0,0} & \sqrt{1-p} \varrho_{0,1} \\
    \sqrt{1-p} \varrho_{1,0} & (1-p) \varrho_{1,1}
    \end{pmatrix}.
  \end{equation}
  In the Kraus representation this channel can be written as~\cite{preskill}
  \begin{eqnarray}
    \label{eq:ADchoi}
    A^{(0)}_{\text{AD}}=\begin{pmatrix}1 & 0\\
      0 & \sqrt{1-p}
    \end{pmatrix},\nonumber \\
    A^{(1)}_{\text{AD}}=\begin{pmatrix}0 & \sqrt{p}\\
      0 & 0\end{pmatrix}.
  \end{eqnarray}
  Now carrying out the vectorization in Eq.~\eqref{eq:Choivec}, and
  using the definition of the Choi matrix in Eq.~\eqref{eq:Choidef},
  we get the Choi matrix of the channel:
  \begin{equation}
    \label{eq:XAD}
    X_{\text{AD}}=
    \begin{pmatrix}1 & 0 & 0 & \sqrt{1-p}\\
      0 & 0 & 0 & 0\\
      0 & 0 & p & 0\\
      \sqrt{1-p} & 0 & 0 & 1-p\end{pmatrix}.
  \end{equation}
  In order to calculate the action of the channel for a density matrix
  $\varrho^{(\text{in})}$, we use Eq.~\eqref{eq:choiact}: first calculate
  the tensor product of the input qubit density matrix and the two-dimensional identity matrix.
then multiply this with $X_{\text{AD}}$ from the
    right, and trace out in the first qubit to get the output:
  \begin{equation}
    \label{eq:ADact}
    \varrho^{\text{(out)}}=\tr_{1}\left(\left(\varrho^{(\text{in})}\otimes \hat 1_2\right)X_{\text{AD}}\right) ,
  \end{equation}
  It is easy to verify that the result coincides with the expected one
  in Eq.~\eqref{eq:adchannel}.

\section{Application: CP map optimization}

As an example application let us recall that one can use the Choi
representation for parametrizing CP maps, i.e. quantum channels that
input quantum states from $\HH_1$ and map them to $\HH_2$.  This is
the parametrization used by Audenaert and De Moor for optimizing
quantum processes~\cite{Audenaert_2002}, and also in one of our
contributions~\cite{1304.1326}.

Let us fix two orthogonal linear bases,
$(\sigma_j)_{j=1\ldots d_1^2-1}$ and $(\tau_k)_{k=1\ldots d_2^2-1}$ in
the linear space of traceless Hermitian matrices in  $d_1$ dimensional Hilbert space $\HH_1$ and $d_2$ dimensional Hilbert space
$\HH_2$; let $\sigma_0$ and $\tau_0$ be the identity operators of the
respective spaces. Then the matrices $(\sigma_i\otimes \tau_k)$ form
an orthogonal basis on $\HH_1\otimes \HH_2$, so we can write an
arbitrary $X: \HH_1\otimes \HH_2 \mapsto \HH_1 \otimes \HH_2$ matrix
as
\begin{equation}
  \label{eq:param}
  X=\sum_{j=0}^{d_1^2-1}\sum_{k=0}^{d_2^2-1}x_{j,k}\sigma_j \otimes \tau_k.
\end{equation}
The condition for preserving trace in Eq.~\eqref{eq:tpcond} implies that
\begin{equation}
  \tr_2 X = \hat 1 =\sigma_0 \leftrightarrow
  \begin{cases}
    x_{j,0}=0&j>0\cr
    x_{0,0}=\frac{1}{d_2}.
\end{cases}
\end{equation}
One can vectorize the $x$ parameters for convenience, use them as
decision variables in an optimization problem. The condition $X\geq 0$
is a semidefinite condition written in the form of a matrix
inequality. Any linear function of the $x$ parameters can thus be
optimized as a linear semidefinite program. Examples include e.g. the
average fidelity of a process to an ideal but nonphysical process one
wants to approximate, as it was done in the aforementioned references.

\section{Product channels}

Multipartite quantum systems often appear in quantum information science. Assume, for instance, that two (possibly separated) parties, Alice and Bob, have access to a part of a bipartite quantum system. The systems may or may not be entangled. As for quantum channels, a reasonable question is the following:
suppose that Alice sends her quantum system through a quantum channel described by a CP map with a Choi matrix $X^{(A)}$, whereas Bob sends his part of the system independently through a channel which is independent of Alice's, and is described by a Choi matrix $X^{(B)}$. This can model, for instance, a Bell-type situation in which the subsystems of an entangled quantum system shared by Alice and Bob are affected by some noise or decoherence independently. We will show that the Choi matrix describing the joint action of the two independent channels on the whole bipartite system is $X^{(A)}\otimes X^{(B)}$.

As for the notation, we will use the labels $A$ and $B$ for the subsystems at Alice and Bob, respectively, whereas $1$ stands for the input space, and $2$ for the output space. Note that there are thus four subsystems in argument, $1A$, $1B$, $2A$, and $2B$.

\begin{proposition}
Given two independent CPTP maps $\EE_A$ and $\EE_B$, represented by Choi matrices $X^{(A)}$ and $X^{(B)}$, respectively, and a bipartite product state as an input
\begin{equation}
\varrho=\varrho_{1A,1B} = \varrho_{1A} \otimes \varrho_{1B},
\end{equation}
for the CPTP map $\EE_{AB}$ represented by the Choi matrix $X_A\otimes X_B$,
\begin{equation}
\EE_{AB}(\varrho)=\EE_{A}(\varrho_{A})\otimes \EE_{B}(\varrho_{B})
\end{equation}
holds.
\end{proposition}


\begin{proof}

The statement can be shown by straightforward direct calculation: writing $X$ with coordinates (using the subsystem labels for the coordinate indices which have been introduced above):
\begin{eqnarray}
X=(X_A \otimes X_B)_{i_{1A},i_{2A};i_{1B},i_{2B}|j_{1A},j_{2A};j_{1B},j_{2B}}  \nonumber \\
=(X_A)_{i_{1A},i_{2A};j_{1A},j_{2A}}(X_B)_{i_{1B},i_{2B};j_{1B},j_{2B}}, 
\end{eqnarray}
and thus we can write the effect of $\EE$ on $\varrho$ as
\begin{eqnarray}
\EE(\varrho) = \tr_{1A,1B}  \left(  (\varrho_{1A,1B} \otimes \hat{1}_{2A,2B}) X^{(AB)} \right)= \varrho'^{(2A,2B)}= \varrho'_{i_{2A},i_{2B};j_{2A},j_{2B}}.
\end{eqnarray}
On the other hand we have
\begin{eqnarray}
\EE_A(\varrho^{(A)}) \otimes \EE_B(\varrho^{(B)}) =  \nonumber \\ \sum_{i_{1A},i_{1B}} \left[ \left( \varrho^\top_{i_{1A},i_{1B}}  \otimes \hat{1}_{2A,2B} \right) X \right]_{i_{1A},i_{2A},i_{1B},i_{2B}|i_{1A},j_{2A},i_{1B},j_{2B}}  \nonumber \\
=\sum_{i_{1A},i_{1B}} \sum_{k_{1A},k_{2A},k_{1B},k_{2B}} \left( \varrho^\top \otimes \hat{1}_{2A,2B} \right)_{i_{1A},i_{2A},i_{1B},i_{2B}|k_{1A},k_{2A},k_{1B},k_{2B}} \times \nonumber\\
 X_{k_{1A},k_{2A},k_{1B},k_{2B}|i_{1A},j_{2A},i_{1B},j_{2B}} \nonumber \\
= \sum_{i_{1A},i_{1B},k_{1A},k_{2A},k_{1B},k_{2B}} \left( \varrho^\top \right)_{i_{1A},i_{1B},k_{1A},k_{1B}} \delta_{i_{2A},k_{2A}} \delta_{i_{2B},k_{2B}} X_{k_{1A},k_{2A},k_{1B},k_{2B}|i_{1A},j_{2A},i_{1B},j_{2B}} \nonumber \\
=  \sum_{i_{1A},i_{1B},k_{1A},k_{2A},k_{1B},k_{2B}} \delta_{i_{2A},k_{2A}} \delta_{i_{2B},k_{2B}} \varrho^{(A)}_{k_{1A},i_{1A}}\varrho^{(B)}_{k_{1B},i_{1B}} X^{(A)}_{k_{1A},k_{2A};i_{1A},j_{2A}} X^{(B)}_{k_{1B},k_{2B};i_{1B},j_{2B}}  \nonumber \\
= \sum_{i_{1A},k_{1A},k_{2A}}\sum_{i_{1B},k_{1B},k_{2B}} \varrho^{(A)}_{k_{1A},i_{1A}} \delta_{i_{2A},k_{2A}} X^{(A)}_{k_{1A},k_{2A};i_{1A},j_{2A}} \varrho^{(B)}_{k_{1B},i_{1B}} \delta_{i_{2B},k_{2B}} X^{(B)}_{k_{1B},k_{2B};i_{1B},j_{2B}} \nonumber \\
= \left( \sum_{i_{1A},k_{1A}}  \varrho^{(A)}_{k_{1A},i_{1A}} X^{(A)}_{k_{1A},i_{2A};i_{1A},j_{2A}} \right) \left( \sum_{i_{1B},k_{1B}}  \varrho^{(B)}_{k_{1B},i_{1B}}X^{(B)}_{k_{1B},i_{2B};i_{1B},j_{2B}} \right) \nonumber \\
=\left[ \left(  \tr_{1A}  \Big\{  \left( \varrho^\top_{1A} \otimes \hat{1}_{2A} \right) X^{(A)} \Big\} \right)  \otimes \left(\tr_{1B} \Big\{ \left( \varrho^\top_{1B} \otimes \hat{1}_{2B} \right) X^{(B)} \Big\} \right) \right]_{i_{2A},i_{2B},j_{2A},j_{2B}}. \nonumber \\
\end{eqnarray}
\end{proof}
The linearity of CPTP maps implies that the bipartite channel arising as the parallel application of independent channels is represented by the direct product of the two channels' Choi matrices.

\section{Discussion and conclusions}

Based on the knowledge of the Kraus representation of CPTP maps we
have derived their Choi representation via basic calculations, and
have studied some of its properties and their implications. These
techniques can be helpful e.g. when analyzing systems which involve
arbitrary quantum channels or when optimizing over general quantum
channels. Apart from this the provided details can help to understand
certain details of completely positive maps.

We also remark here that throughout this work we have used the term
'completely positive' in the sense of linear operators mapping
Hermitian operators to Hermitian operators, and their positivity being
related to preservation of positive semidefiniteness. To avoid
confusion it should be noted that 'completely positivity' of linear
operators can have another meaning: especially in the literature of
mathematical optimization completely positive linear operators (on a
real vector space) are defined~\cite{Berman_2003} in the element-wise
positivity sense as
\begin{equation}
  \label{eq:mathCP}
  \mathcal{O}=\sum_k x^{(k)}x^{(k)\top},\quad x_k\in \mathbb{R}_+^n,
\end{equation}
i.e. as a convex combination of outer products of real vectors of the
positive orthant. These completely positive operators  form a cone, and
they have many applications~\cite{Berman_2003,
  10.1007/978-3-642-12598-0_1}, with implications to
quadratic and combinatorial optimization. However, there is no
connection known between this latter notion of 'complete positivity' and the
one discussed in the present tutorial.

\section*{Author contribution}

All the authors have accepted responsibility for the entire content of
this submitted manuscript and approved submission.

\section*{Research funding}

This research was supported by the National Research, Development, and Innovation Office of Hungary under the "Frontline" Research Excellence Program, (Grant. No. KKP 133827), and Project no. TKP2021-NVA-04. This project has received funding from the European Union under grant agreement No. 101081247 (QCIHungary project) and has been implemented with the support provided by the Ministry of Culture and Innovation of Hungary from the National Research, Development and Innovation Fund. 
A. O. acknowledges support from the program 'Apoyos para la
Incorporación de Investigadoras e Investigadores Vinculada a la
Consolidación Institucional de Grupos de Investigación 2023' from
CONAHCYT, Mexico.

\section*{Conflict of interest statement}
The authors declare no conflicts of interest regarding this article.

\section*{Acknowledgements}

The authors are indebted to A.~Frigyik, A.~C.~Reynoso and \'A.~T\'oth
for helpful discussions.


\begin{thebibliography}{10}

\bibitem{Audenaert_2002}
K.~Audenaert and B.~D. Moor.
\newblock Optimizing completely positive maps using semidefinite programming.
\newblock {\em Phys. Rev. A}, 65(3):030302, feb 2002.

\bibitem{Benentibook2019}
G.~Benenti, G.~Casati, D.~Rossini, and G.~Strini.
\newblock {\em Principles of Quantum Computation and Information}.
\newblock World Scientific, 1st edition, 2019.

\bibitem{Berman_2003}
A.~Berman and N.~Shaked-Monderer.
\newblock {\em Completely Positive Matrices}.
\newblock World Scientific, Apr. 2003.

\bibitem{Breuerbook2007}
H.-P. Breuer and F.~Petruccione.
\newblock {\em The Theory of Open Quantum Systems}.
\newblock Oxford University Press, 1st edition, 2007.

\bibitem{Choi_1975}
M.-D. Choi.
\newblock Completely positive linear maps on complex matrices.
\newblock {\em Linear Algebra and its Applications}, 10(3):285--290, jun 1975.

\bibitem{Daviesbook1976}
E.~B. Davies.
\newblock {\em Quantum Theory of Open Systems}.
\newblock Academic Press Inc., 1st edition, 1976.

\bibitem{10.1007/978-3-642-12598-0_1}
M.~D{\"u}r.
\newblock Copositive programming -- a survey.
\newblock In M.~Diehl, F.~Glineur, E.~Jarlebring, and W.~Michiels, editors,
  {\em Recent Advances in Optimization and its Applications in Engineering},
  pages 3--20, Berlin, Heidelberg, 2010. Springer Berlin Heidelberg.

\bibitem{2211.16533}
M.~Frembs and E.~G. Cavalcanti.
\newblock Variations on the choi-jamiolkowski isomorphism, 2022.

\bibitem{Heinosaaribook2012}
T.~Heinosaari and M.~Ziman.
\newblock {\em The Mathematical Language of Quantum Theory}.
\newblock Cambridge University Pres, 1st edition, 2012.

\bibitem{Holevobook2019}
A.~S. Holevo.
\newblock {\em Quantum Systems, Channels, Information: A Mathematical
  Introduction}.
\newblock De Gruyter, exp ed. edition, 2019.

\bibitem{Jamiolkowski_1972}
A.~Jamiołkowski.
\newblock Linear transformations which preserve trace and positive
  semidefiniteness of operators.
\newblock {\em Reports on Mathematical Physics}, 3(4):275–278, Dec. 1972.

\bibitem{1304.1326}
M.~Koniorczyk, L.~Dani, and V.~Bužek.
\newblock Process optimized quantum cloners via semidefinite programming, 2013.
\newblock arXiv:1304.1326.

\bibitem{Krausbook1983}
K.~Kraus.
\newblock {\em States, Effects, and Operations: Fundamental Notions of Quantum
  Theory}.
\newblock Lecture Notes in Physics, 190, Springer, 1st edition, 1983.

\bibitem{maxima}
Maxima.
\newblock Maxima, a computer algebra system. version 5.47.0, 2023.
\newblock \url{https://maxima.sourceforge.io/}.

\bibitem{Nielsenbook2010}
M.~A. Nielsen and I.~L. Chuang.
\newblock {\em Quantum Computation and Quantum Information}.
\newblock Cambridge University Press, 10th aniversary edition edition, 2020.

\bibitem{preskill}
J.~Preskill.
\newblock Lecture notes for physics 219/computer science 219: Quantum
  computation, chapter 3.: Measurement and evolution, July 2015.
\newblock \url{http://www.theory.caltech.edu/~preskill/ph219/chap3_15.pdf},
  visted 2024.06.25.

\bibitem{Royer_1991}
A.~Royer.
\newblock Wigner function in {Liouville} space: {A} canonical formalism.
\newblock {\em Physical Review A}, 43(1):44–56, Jan. 1991.

\bibitem{Neumannbook1955}
J.~von Neumann.
\newblock {\em Mathematical Foundations of Quantum Mechanics}.
\newblock Princeton University Press, 1st edition, 1955.

\bibitem{Watrousbook2018}
J.~Watrous.
\newblock {\em The Theory of Quantum Information}.
\newblock Cambridge University Press, 1st edition, 2018.

\bibitem{Wisemanbook2009}
H.~M. Wiseman and G.~J. Milburn.
\newblock {\em Quantum Measurement and Control}.
\newblock Cambridge University Press, 1st edition, 2009.

\end{thebibliography}

\end{document}